\documentclass[a4paper]{article}

\usepackage[english]{babel}
\usepackage[utf8]{inputenc}
\usepackage{amsmath}
\usepackage{graphicx}
\usepackage[colorinlistoftodos]{todonotes}
\usepackage[a4paper,left=2.5cm,right=2.5cm,top=2.cm,bottom=2cm]{geometry}
\usepackage{authblk}
\usepackage{txfonts}
\usepackage{bbold} 
\usepackage{amsopn} 
\usepackage{graphicx}
\usepackage{dsfont}
\newtheorem{example}{Example}
\newtheorem{proposition}{Proposition}
\newtheorem{definition}{Definition}
\newtheorem{lemma}{Lemma}
\newtheorem{remark}{Remark}

\newenvironment{proof}{\noindent\textit{Proof~~}}
{\nolinebreak[4]\hfill$\blacksquare$\\\par}
\title{Quantum Penny Flip game with unawareness}

\author{Piotr Fr\c{a}ckiewicz}
\affil{Institute of Mathematics\\ Pomeranian University, Poland}
\date{\today}

\begin{document}
\maketitle

\begin{abstract}
Games with unawareness model strategic situations in which players' perceptions about the game are limited. They take into account the fact that the players may be unaware of some of the strategies available to them or their opponents as well as the players may have a restricted view about the number of players participating in the game. The aim of the research is to introduce this notion into theory of quantum games. We shall focus on PQ Penny Flip game introduced by D. Meyer. We shall formalize the previous results and consider other cases of unawareness in the game.
\end{abstract}
\section{Introduction}
\label{intro}
Game theory, launched in 1928 by John von Neumann in a paper \cite{neumann1} and developed in 1944 by John von Neumann and Oskar Morgenstern in a book \cite{neumann2} is one of the youngest branches of mathematics. The aim of this theory is mathematical modeling of behavior of rational participants of conflict situations who aim at maximizing their own gain and take into account all possible ways of behaving of remaining participants. Within this young theory new ideas that improve already used models of conflict situations are still proposed. One of the latest trends is to study games with unawareness, i.e., games that describe situations in which a player behaves according to his own view of the game, and considers how all the remaining players view the game. This way of describing a conflict situation goes beyond the most frequently used paradigm, according to which it is assumed that all participants in a game have full knowledge of the situation.

The other, equally young field developed on the border of game theory and quantum information theory is quantum game theory. This is an interdisciplinary area of research within which considered games are supposed to be played with the aid of objects that behave according to the laws of quantum mechanics, and in which non-classical features of these objects are relevant to the way of playing and results of a game.

Game with unawareness is a relatively new notion. The first attempts at formalizing that concept can be found in papers \cite{halpern} and \cite{feinbergstary} published already in XXI century. In paper \cite{feinberg} there is a summary of results obtained in this area till 2012. Quantum counterparts of games with unawareness have not been studied yet. Papers on quantum games with incomplete information concerned only Bayesian games \cite{han}, \cite{iqbalbayesian}, \cite{situ}, \cite{situ2} and games with imperfect recall \cite{cabello}, \cite{fracor1}. Our project is the first attempt to use the notion of game with unawareness in theory of quantum games. 
The main motivation for our interest in developing this branch of quantum game theory was our observation that already in the first paper on quantum games by D. Meyer \cite{meyer} its author unconsciously utilized the idea of game with unawareness. We may conclude from the famous PQ Penny Flip game described in \cite{meyer} that Captain Picard (player 2) agrees to join the game because his chance of winning is 1/2. In other words, the game he perceives is the classical one.  Q (player 1) views the game in a different way. He is aware of unitary strategies. In addition, player 1 knows that player 2 is only aware of the classical strategies. This knowledge is crucial in the way he chooses his strategy. Choosing, for example, the Hadamard matrix always leads player 1 to getting the best possible outcome. It is optimal to player 1 to play that strategy since he is aware that player 2 has no counter strategies available. 
Once we learned the quantum PQ Penny Flip game is a game with unawareness, the description of the game, say by using normal form, requires a family of games rather than a single normal-form game. This has numerous important consequences in the form of solution concepts supposed to predict rational results of the game. In particular, Nash equilibrium concept is not sufficient to fully describe the players' rational choices. In the case of PQ Penny Flip game in which quantum strategies are available only for player 1, each player 2's mixed strategy is an equilibrium strategy. However, taking into account player 2's view about the game (he finds the game to be the classical one), we should predict that he chooses his pure strategies with equal probability. 
\section{PQ Penny Flip game}
\label{sec:1}
Formally, the classically played PQ Penny Flip game \cite{meyer} is an example of a two-person extensive-form game whose the game tree is depicted at the top of Fig.\ref{fig:1}. 
\begin{figure}
  \includegraphics[width=1.0\textwidth]{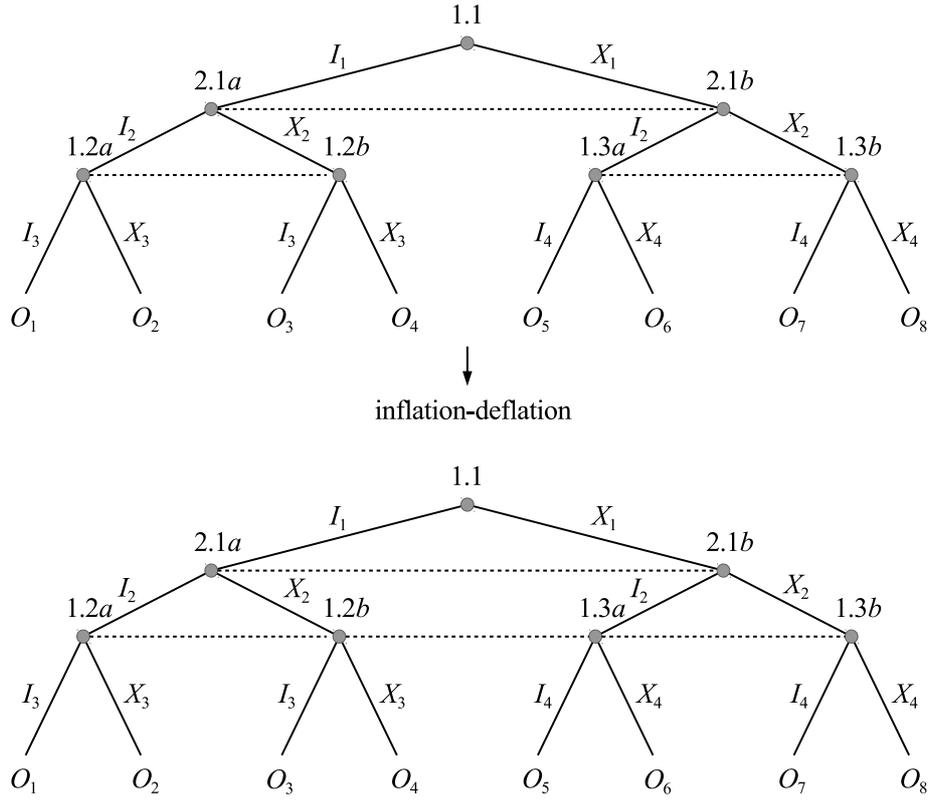}
\caption{General extensive form of the classical PQ Penny Flip game (the top figure) and its equivalent counterpart according to the inflation-deflation principle (the bottom figure).}
\label{fig:1}       
\end{figure}
Player~1 initiates the game by choosing one of the two available actions $I_{1}$ and $X_{1}$. Next, player 2 chooses a possible action in $\{I_{2}, X_{2}\}$. The dashed line connecting vertices 2.1a and 2.1b in Fig.~\ref{fig:1} indicates player 2's two-element information set. This means that player~2 does not know whether player 1 has chosen $I_{1}$ or $X_{1}$. Similarly, player 1 does not know a move made by his predecessor at the time he chooses his second action. As a result, player 1 has two two-element information sets $\{1.2a, 1.2b\}$, $\{1.3a, 1.3b\}$ and one-element information set $\{1.1\}$.

Every extensive-form game can be associated with a strategic-form game. The latter form is particularly convenient when the two-person extensive game is to be studied with respect to Nash equilibria. A strategic-form game is derived from an extensive form-game by determining the set of strategies $S_{i}$ of each player $i$, and the payoffs induced by all strategy profiles in the extensive-form game. A strategy of player $i$ is a function mapping each of her information sets to an element in the set of actions at that information set (see for example \cite{maschler}). In the case of the top game in Fig.~\ref{fig:1} a player 1's strategy is an element of $\{I_{1}, X_{1}\} \times \{I_{3}, X_{3}\} \times \{I_{4}, X_{4}\}$. Hence, the strategic form of the extensive game in Fig.~\ref{fig:1} and its reduced form is as follows:
\begin{equation}
\bordermatrix{& I_{2} & X_{2} \cr 
I_{1}I_{3}I_{4} & O_{1} & O_{3} \cr 
I_{1}I_{3}X_{4} & O_{1} & O_{3} \cr 
I_{1}X_{3}I_{4} & O_{2} & O_{4} \cr
I_{1}X_{3}X_{4} & O_{2} & O_{4} \cr
X_{1}I_{3}I_{4} & O_{5} & O_{7} \cr
X_{1}I_{3}X_{4} & O_{6} & O_{8} \cr 
X_{1}X_{3}I_{4} & O_{5} & O_{7} \cr
X_{1}X_{3}X_{4} & O_{6} & O_{8}}, \quad \bordermatrix{& I_{2} & X_{2} \cr
\{I_{1}I_{3}I_{4}, I_{1}I_{3}X_{4}\} &O_{1} & O_{3} \cr 
\{I_{1}X_{3}I_{4}, I_{1}X_{3}X_{4}\} & O_{2} & O_{4} \cr 
\{X_{1}I_{3}I_{4}, X_{1}X_{3}I_{4}\} & O_{5} & O_{7} \cr
\{X_{1}I_{3}X_{4}, X_{1}X_{3}X_{4}\} & O_{6} & O_{8}}.
\end{equation}
We see at once that player 1 has four strategies that are equivalent to the other four ones (they generate the same outcomes). A strategic-form game in which every set of equivalent strategies is replaced by a single strategy from that set is called a game in reduced strategic form. Hence, the extensive game at the top of Fig.~\ref{fig:1} can be associated with $4\times 2$ reduced strategic form. In this case, we can identify player 1' strategy as a map that specifies one action in $\{1.1\}$ and one action in the union of information sets $\{1.2a, 1.2b\}$ and $\{1.3a, 1.3b\}$. In other words, the meaningful player 1' strategies may be written as $(a_{1}, a_{3})$, where $a_{1}$ and $a_{3}$ are actions taken at the first and the third stage of the game, respectively. It is worth noting that it is still holds if the cardinality of the sets of players' actions is greater than 2. This property will be used throughout the paper and it follows from one of the four transformations preserving the reduced strategic form called {\it Inflation-Deflation} (see \cite{thomson} and \cite{osborne}).
\paragraph{Inflation-Deflation} The extensive games $\Gamma$ and $\Gamma'$ share the same reduced strategic-form game if $\Gamma'$ differs from $\Gamma$ only in an information set of some player $i$ in $\Gamma$ that is a union of information sets of player $i$ in $\Gamma'$ ($\{1.2a, 1.2b\}$ and $\{1.3a, 1.3b\}$ in Fig.~\ref{fig:1}) with the following property: any two sequences of actions $h$ and $h'$ leading from the root of the game tree to different members of the union (for example, sequences $(I_{1}, X_{2})$ and $(X_{1}, X_{2})$) have the subsequences that lead to the same information set of player $i$ (empty sequence $\emptyset$ in our case) and player $i$'s action at this information set is different in $h$ and $h'$.
\vspace{12pt}

As it was mentioned at the beginning of this section, the classical Penny Flip game \cite{meyer} is a special case of the game in Fig.~\ref{fig:1}. It is obtained by setting
\begin{equation}
O_{1} = O_{4} = O_{6} = O_{7} = (1,-1), \quad O_{2} = O_{3} = O_{5} = O_{8} = (-1,1).
\end{equation}
On account of the inflation-deflation principle we may write the strategic-form game as
\begin{equation}\label{classic}
\bordermatrix{& I_{2} & X_{2}\cr 
I_{1}I_{3} & (1,-1) & (-1,1)\cr 
I_{1}X_{3} & (-1,1) & (1,-1)\cr 
X_{1}I_{3} & (-1,1) & (1,-1) \cr
X_{1}X_{3} & (1,-1) & (-1,1)}.
\end{equation}
One can check that mixed strategies defined by probability distributions
\begin{equation}\label{cqpf}
\left(\frac{1}{2}, \frac{1}{2}, 0,0\right), \left(\frac{1}{2}, 0, \frac{1}{2},0\right), \left(0,0,\frac{1}{2}, \frac{1}{2}\right), \left(0,\frac{1}{2}, 0, \frac{1}{2}\right)
\end{equation}
over the set $\{I_{1}I_{3}, I_{1}X_{3}, X_{1}I_{3}, X_{1}X_{3}\}$ are the optimal strategies for player 1 in game~(\ref{classic}), and, thus, also each probability distribution over~(\ref{cqpf}). The optimal strategy for player 2 is, in turn, determined by the unique probability distribution $(\frac{1}{2}, \frac{1}{2})$ over $\{I_{2}, X_{2}\}$. Hence, the value of game (\ref{classic}) is equal to zero. 

Meyer \cite{meyer} generalized the PQ Penny Flip game by using quantum computing formalism. The general strategic-form of the game $(N, (S_{i})_{i\in N}, (u_{i})_{i\in N})$ in which both players have access to unitary strategies can be written formally as
\begin{equation}\label{qq}
\Gamma_{QQ} = \left(\{1,2\}, \{(U_{1}, U_{3})\}, \{U_{2}\}, \{\mathrm{tr}(\rho_{\mathrm{f}}P), -\mathrm{tr}(\rho_{\mathrm{f}}P)\}\right),
\end{equation}
where
\begin{itemize}
\item $\{1,2\}$ is a set of players,
\item $(U_{1}, U_{3})$ and $U_{2}$ are strategies of player 1 and 2, respectively, and $U_{j}$ is a $2\times 2$ unitary matrix for each $j$,
\item  $\rho_{\mathrm{f}}$ is a density matrix defined as follows
\begin{equation}
\rho_{\mathrm{f}} = U_{3}U_{2}U_{1}|0\rangle \langle 0| U^{\dagger}_{1}U^{\dagger}_{2}U^{\dagger}_{3},
\end{equation} 
\item $P$ is a Hermitian operator in the form
\begin{equation}\label{hermitianp}
P = |0\rangle \langle 0| - |1\rangle \langle 1|.
\end{equation}
\end{itemize}
Let us denote by $\mathds{1}$ the identity matrix of size 2, and by $\sigma_{i}$, $i=x,y,z$, the Pauli matrix $i$.
It follows easily that game~(\ref{classic}) is a special case of (\ref{qq}), if the set of unitary actions $U_{j}$ is restricted to the set $\{\mathds{1}, \sigma_{x}\}$. 

We shall use the following notation for the PQ Penny Flip game with unitary actions restricted to $\{\mathds{1}, \sigma_{x}\}$:
\begin{align}\label{notation}
\begin{split}
&\Gamma_{CC} = \left(\{1,2\}, \{(\mathds{1}, \mathds{1}), (\mathds{1}, \sigma_{x}), (\sigma_{x}, \mathds{1}), (\sigma_{x}, \sigma_{x})\}, \{\mathds{1}, \sigma_{x}\}, \{\mathrm{tr}(\rho_{\mathrm{f}}P), -\mathrm{tr}(\rho_{\mathrm{f}}P)\}\right),\\
&\Gamma_{QC} = \left(\{1,2\}, \{(U_{1}, U_{3})\}, \{\mathds{1}, \sigma_{x}\}, \{\mathrm{tr}(\rho_{\mathrm{f}}P), -\mathrm{tr}(\rho_{\mathrm{f}}P)\}\right),\\
&\Gamma_{CQ} = \left(\{1,2\}, \{\mathds{1}\mathds{1}, \mathds{1}\sigma_{x}, \sigma_{x}\mathds{1}, \sigma_{x}\sigma_{x}\}, \{U_{2}\}, \{\mathrm{tr}(\rho_{\mathrm{f}}P), -\mathrm{tr}(\rho_{\mathrm{f}}P)\}\right).
\end{split}
\end{align}
For example, in the game $\Gamma_{CQ}$ player 1 is restricted to use only classical actions whereas player 2's set of actions is the set of $2\times 2$ unitary matrices.

One of the main ideas behind the PQ Penny Flip game was to show that Alice can win the game every time she plays against Bob. It is possible if Alice has access to unitary strategies that Bob is not aware of. That is, Alice is fully aware of unitary operators available in the quantum PQ Penny Flip game, Bob is only aware of unitary operations identified with his strategies in the classical PQ Penny Flip game (for example, $\mathds{1}$ and $\sigma_{x}$).

The common example of Alice's winning strategy is playing the Hadamard matrix $H$ twice: 
\begin{equation}\label{awexample}
|0\rangle \xrightarrow[\text{Alice}]{H} \frac{1}{\sqrt{2}}(|0\rangle + |1\rangle) \xrightarrow[\text{Bob}]{\mathds{1}~\text{or}~\sigma_{x}} \frac{1}{\sqrt{2}}(|0\rangle + |1\rangle) \xrightarrow[\text{Alice}]{H} |0\rangle.
\end{equation}
Starting from the state $|0\rangle$ (which is identified with the coin heads up), Alice creates the equal superposition state $(|0\rangle + |1\rangle)/\sqrt{2}$ by using the operator $H$. Bob having only $\mathds{1}$ and $\sigma_{x}$ cannot affect that superposition state. For that reason, Alice again chooses the Hadamard matrix and then she gets the state $|0\rangle$ back. From (\ref{awexample}) it follows  that for any Bob's mixed strategy $(p,1-p)$ over $\{\mathds{1}, \sigma_{x}\}$ Alice wins the game by playing $HH$, i.e.,
\begin{equation}\label{mixedaw}
\operatorname{tr}\left(\left(pH\mathds{1}H|0\rangle \langle 0|H\mathds{1}H + (1-p)H\sigma_{x}H|0\rangle \langle 0|H\sigma_{x}H\right)P\right)  = \operatorname{tr}(|0\rangle \langle 0|P)= 1.
\end{equation}
\subsection{Technical difficulties in describing PQ Penny Flip problem}\label{techsection}
We know from (\ref{mixedaw}) that Alice can win the PQ Penny Flip game if she has access to the Hadamard matrix, and Bob is not aware of unitary matrices except $\mathds{1}$ and $\sigma_{x}$. A natural question arises as to how this problem can be described from a game theory point of view.

At first glance, the following strategic-form game seems to express that problem:
\begin{equation}\label{misleading}
\bordermatrix{& \mathds{1} & \sigma_{x} \cr
\mathds{1}\mathds{1} & (1,-1) & (-1,1) \cr 
\mathds{1}\sigma_{x} & (-1,1) & (1,-1) \cr
\sigma_{x}\mathds{1} & (-1,1) & (1,-1)\cr 
\sigma_{x}\sigma_{x} & (1,-1) & (1,-1) \cr
\mathds{1}H & (0,0) & (0,0) \cr   
\sigma_{x}H & (0,0) & (0,0) \cr 
H\mathds{1} & (0,0) & (0,0) \cr 
H\sigma_{x} & (0,0) & (0,0) \cr 
HH & (1,-1) & (1,-1)}.
\end{equation}
We infer from bimatrix game (\ref{misleading}) that Alice has an additional move $H$ compared with her actions in the classical PQ Penny Flip game, and now she has eight strategies. Bob has two strategies equivalent to ones in the classical game. Moreover, looking at (\ref{misleading}) we see that $HH$ is Alice's winning strategy. 

Up to now, (\ref{misleading}) appears to agree with the PQ Penny Flip problem.  In fact, (\ref{misleading}) turns out to provide Bob with much richer description of the game than he actually has. Making his strategic decision based on (\ref{misleading}), Bob finds that Alice has the additional action $H$, and consequently the winning strategy. Perhaps, Bob does not know that $H$ is the Hadamard matrix or he does not even realize that he is to play the quantum game. However, Bob knows that he looses the game. According to \cite{meyer}, Bob agrees to play the PQ Penny Flip game because he is confident that the odds of winning the game are even, and his optimal strategy is to play $\mathds{1}$ and $\sigma_{x}$ with equal probability. In the case of (\ref{misleading}) Bob gets the payoff of -1, no matter which strategy he chooses.  Therefore, Bob's optimal strategy in (\ref{misleading}) is any probability distribution over his set of strategies. 

The solution is to consider a family of games--the core of the definition of games with unawareness. The formal definition can take into account a player's view about his strategy set or strategies of the other players, a player's view about other players' views, and even a player's view about the number of players taking part in the game.
\section{Preliminaries on games with unawareness}
For the convenience of the reader we review the relevant material from \cite{feinberg}. Before we begin the formal presentation, we will look at an example that illustrates that concept and the ideas behind it. The reader who is not familiar with this topic is encouraged to see a similar introductory example in \cite{feinberg}.
\begin{example}\label{introductoryexample}
Let us consider the following bimatrix game
\begin{equation}\label{bimatrix1}
\Gamma_{1}\colon ~\bordermatrix{& b_{1} & b_{2} \cr 
a_{1} & (2,2) & (2,2) \cr 
a_{2} & (3,3) & (1,2) \cr
a_{3} & (4,0) & (1,2)
}.
\end{equation}
We assume that Alice (player 1) and Bob (player 2) are both aware of all the strategies available in game (\ref{bimatrix1}). However, we consider the situation where Bob finds that Alice views the game in the following form:
\begin{equation}\label{bimatrix2}
\Gamma_{2}\colon \bordermatrix{& b_{1} & b_{2} \cr 
a_{1} & (2,2) & (2,2) \cr 
a_{2} & (3,3) & (1,2) 
}.
\end{equation}
In words, Bob perceives Alice's strategy set to be $\{a_{1}, a_{2}, a_{3}\}$, but for some reason, she thinks that Alice views $\{a_{1}, a_{2}\}$. Since Bob finds that Alice views the game being played as depicted in (\ref{bimatrix2}), Bob thinks that Alice finds that he also considers (\ref{bimatrix2}), and so on for higher order views.

Let us consider the case that Alice is fully aware of Bob's reasoning. Not only does she perceive her whole strategy set $\{a_{1}, a_{2}, a_{3}\}$, Alice also finds that Bob does not realize that she is considering $\{a_{1}, a_{2}, a_{3}\}$ but $\{a_{1}, a_{2}\}$. Moreover, Alice finds that Bob views the game as in (\ref{bimatrix1}). 

The problem just presented is an example of a strategic-form game with unawareness that can be formally described by a family of games $\{G_{v}\}_{v\in \mathcal{V}_{0}}$, where $\mathcal{V}_{0} = \{\emptyset, 1, 2, 12, 121, \dots\}$, and
\begin{equation}\label{formula11}
G_{v} = \begin{cases}
\Gamma_{1} &\text{if}~v\in \{\emptyset, 1, 2, 12\}, \\
\Gamma_{2} &\text{otherwise}.
\end{cases}
\end{equation}
The set $\mathcal{V}_{0}$ (with typical element $v$) consists of the relevant views. The view $v=\emptyset$ corresponds to the modeler's game--the actual game played by the players. In our example, this is game (\ref{bimatrix1}). That game is also viewed by player 1 ($v=1$) and player 2 ($v=2$). Furthermore, according to the description of the game, player 1 (Alice) finds that player 2 (Bob) is considering $\Gamma_{1}$. It is taken into account in (\ref{formula11}) by associating $\Gamma_{1}$ with the view $v =12$ (the view that player 1 finds that player 2 is considering...). In our example, player 2 finds that player 1 views the game as in (\ref{bimatrix2}). For this reason, the game $\Gamma_{2}$ corresponds to $v=21$. Any higher order iteration of awareness of Alice and Bob are also assumed to be associated with $\Gamma_{2}$. 

We see at once that game (\ref{bimatrix1}) has the unique pure Nash equilibrium $(a_{1}, b_{2})$, and we could check that, in general, the set of all (mixed) Nash equilibria in (\ref{bimatrix1}) is
\begin{equation}\label{setnash}
\left\{(a_{1}, (q,1-q))\colon q \in \left[0, \frac{1}{3}\right]\right\},
\end{equation}
where $(q,1-q)$ denotes player 2's mixed strategy under which he chooses $b_{1}$ and $b_{2}$ with probability $q$ and $1-q$, respectively. Each of the strategy profiles from (\ref{setnash}) yields the payoff outcome $(2,2)$. Although, both players are aware of playing (\ref{bimatrix1}), it is not evident that the game ends with outcome $(2,2)$. According to (\ref{formula11}), Bob finds that Alice perceives game (\ref{bimatrix2}). Hence, he may deduce that Alice plays according to strategy profile $(a_{2}, b_{1})$ as being the most profitable Nash equilibrium in (\ref{bimatrix2}). Bob's choice would be then $b_{1}$. Alice, however, is aware of Bob's thinking. She finds that Bob is considering (\ref{bimatrix1}), and also finds that Bob finds that she is considering (\ref{bimatrix2}). Alice can therefore deduce that Bob chooses strategy $b_{1}$ that weakly dominates $b_{2}$ in (\ref{bimatrix2}), i.e., it gives Bob a payoff at least as high as $b_{2}$, and at the same time, is an element of the most beneficial Nash equilibrium in (\ref{bimatrix2}). Since Alice is aware of playing $\Gamma_{1}$, it is not optimal for her to play according to $(a_{2}, b_{1})$ but to choose $a_{3}$. As a result, the game described above ends with payoff outcome (4,0) corresponding to the strategy profile $(a_{3}, b_{1})$. 

The game result $(a_{3}, b_{1})$ can be directly determined by the extended Nash equilibrium \cite{feinberg} - a solution concept being a counterpart of Nash equilibrium in games with unawareness. The formal definition is presented in Subsection~\ref{enesubsection}. Here we simply provide the result of applying the extended Nash equilibrium to (\ref{formula11}). One of the equilibrium solutions is a family of strategy profiles $((\sigma)_{v})_{v\in \mathcal{V}_{0}}$ defined as follows:
\begin{equation}\label{firstene}
\sigma_{v} = \begin{cases}
(a_{3}, b_{1}) &\text{if}~v\in \{\emptyset,  1\}, \\
(a_{2}, b_{1}) &\text{otherwise}.
\end{cases}
\end{equation}
The strategy profiles (\ref{firstene}) coincide with the reasoning we already used to determine the outcome $(a_{3}, b_{1})$. The result of the game corresponds to the modeler's view $(v=\emptyset)$. It also coincides with Alice's view $(v=1)$ as Alice is fully aware of the games played by her and Bob. The strategy profile $(a_{2}, b_{1})$ is seen from Bob's point of view $(v=2)$. Since Alice is aware of Bob's thinking, she finds that Bob is considering $(a_{2}, b_{1})$ $(v=12)$. 
\end{example}
\subsection{The role of the notion of games with unawareness in quantum game theory}
The notion of games with unawareness is designed to model game theory problems in which players' perceptions of the game are restricted. It was shown in \cite{feinberg} that the novel structure extends the existing forms of games. Although it is possible to represent games with unawareness with the use of games with incomplete information, the extended Nash equilibrium does not map to any known solution concept of incomplete information games. In particular, the set of extended Nash equilibria forms a strict subset of the Bayesian Nash equilibria. 

Once we know that games with unawareness is a new game form, it is natural to study that type of games in the quantum domain. Having given a quantum game scheme that maps a classical game $G$ to the quantum one $Q(G)$, and having given a family of games $\{G_v\}$, a family of quantum games ${Q(G_v)}$ can be constructed in a natural way. Then we can study if, and to what extent, quantum strategies compensate restricted perception of players.

Besides $\{Q(G_{v})\}$, the notion of game with unawareness allows one to expand the theory of quantum games by defining a family $\{Q(G)_{v}\}$, where each quantum game $Q(G)_{v}$ corresponds to a specific perception of players. In this case players may have restricted perception of how a quantum game is defined. A good example of that quantum game theory problem is the quantum PQ Penny Flip game \cite{meyer}: one of the players is aware of having all the quantum strategies, the other player perceives two unitary strategies identified with the classical Penny Flip game. We provide a detail exposition of that problem in Section~\ref{section4}. 

Another example of applying the notion of games with unawareness concerns the case when playing a quantum game is not common knowledge among the players. The quantum game is to be played with the aid of object that behave according to the laws of quantum mechanics, in particular, the players may share an entangled two-qubit state on which they apply unitary strategies. Under this scenario (see figure~\ref{fig:00}), Alice and Bob can be far apart, and a third party, say a modeler, is to prepare the game. After the modeler prepares the quantum game based on its classical counterpart, he sends the message to Alice and Bob so that they know they are to play the quantum game rather that the classical one. When the players receive the message, they each perceive the game as being quantum, i.e., $G_{i} = \Gamma_{Q}$. But this fact is not common knowledge among Alice and Bob. Recall that a fact is common knowledge among the players of a game if for any finite sequence of players $i_{1}, i_{2}, \dots, i_{k}$ player $i_{1}$ knows that player $i_{2}$ knows \dots that player $i_{k}$ knows the fact. In our case, each of the players cannot be certain that the other player finds the quantum game (receives the message from the modeler) until he or she receives a confirmation from that player. According to the scheme in Fig.~\ref{fig:00}, Alice sends   Bob a message about her current state of knowledge. In this way, Bob receiving the message finds that Alice is considering the quantum game, i.e., $G_{21} = \Gamma_{Q}$. Now, Bob sends the feedback message including his state of knowledge.   Owing to this message, Alice finds that Bob is considering the quantum game, $G_{12} = \Gamma_{Q}$. Moreover, Alice finds that Bob finds that Alice is considering the quantum game, $G_{121} = \Gamma_{Q}$. At this point, the quantum game is still not considered common knowledge. Bob is not certain that Alice finds that Bob is considering the quantum game until he receives the second message from Alice. Since the game starts before the message arrives at Bob, at the time of the play, either the classical game $\Gamma_{C}$ or the quantum game $\Gamma_{Q}$ may be associated with $G_{212}$, and the same conclusion can be drawn for the higher levels of views. As a result, the players face a game with unawareness described by a family of games $\{G_{v}\}$ rather than the single game $\Gamma_{Q}$. An example of the game being in line with the scheme in Fig.~\ref{fig:00} is a family $\{G_{v}\}_{v\in \mathcal{V}_{0}}$, where
\begin{equation}
G_{v} = \begin{cases} 
\Gamma_{Q} &\text{if}~v\in \{\emptyset, 1, 2, 12, 21, 121\}, \\
\Gamma_{C} &\text{otherwise}.
\end{cases}
\end{equation}
\begin{figure}
\centering
  \includegraphics[width=0.65\textwidth]{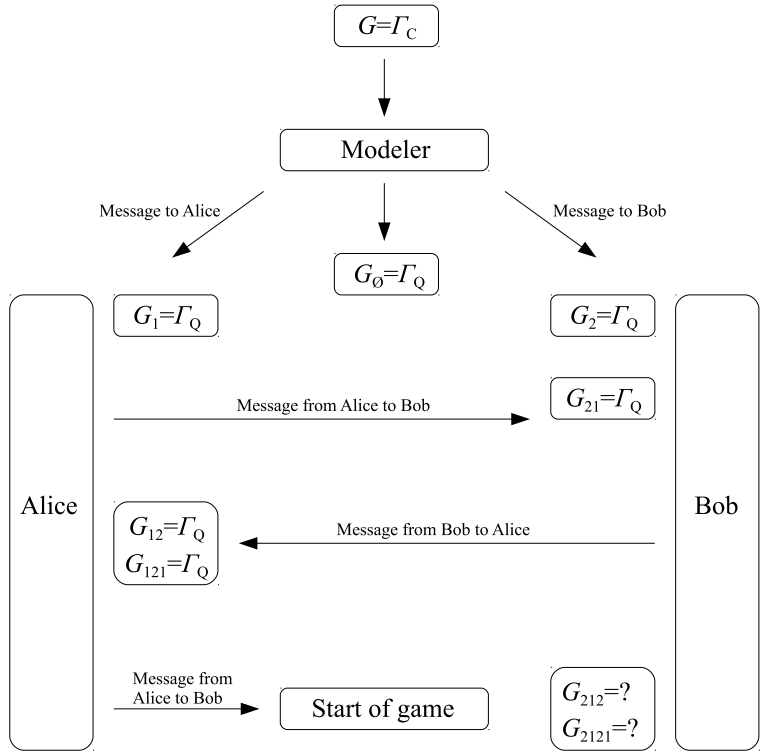}
\caption{A possible scenario before a quantum game is played}
\label{fig:00}       
\end{figure}
We will later see that the result of the game differs significantly depending on how the players perceive the game.
\subsection{Strategic-form games with unawareness}
Let $G = \left(N, \prod_{i\in N}S_{i}, (u_{i})_{i\in N}\right)$ be a strategic form game. This is the game considered by the modeler. Each player may not be aware of the full description of $G$.  Hence $G_{\mathrm{v}} = \left(N_{\mathrm{v}}, \prod_{i\in N_{\mathrm{v}}}(S_{i})_{\mathrm{v}}, ((u_{i})_{\mathrm{v}})_{i\in N_{\mathrm{v}}}\right)$ denotes player $\mathrm{v}$'s view of the game for $\mathrm{v} \in N$. In general, each player also considers how each of the other players views the game. Formally, with a finite sequence of players $v = (i_{1}, \dots, i_{n})$ there is associated a game $G_{v} = \left(N_{v}, \prod_{i\in N_{v}}(S_{i})_{v}, ((u_{i})_{v})_{i\in N_{v}}\right)$. This is the game that player $i_{1}$ considers that player $i_{2}$ considers that \dots player $i_{n}$ is considering. A sequence $v$ is called a view. The empty sequence $v = \emptyset$ is assumed to be the modeler's view ($G_{\emptyset} = G$). We denote a strategy profile in $G_{v}$ by $(s)_{v}$. The concatenation of two views $\bar{v} = (i_{1}, \dots, i_{n})$ followed by $\tilde{v} = (j_{1}, \dots, j_{n})$ is defined to be $v = \bar{v}\string^\tilde{v} = (i_{1}, \dots, i_{n}, j_{1}, \dots, j_{n})$. The set of all potential views is $V = \bigcup^{\infty}_{n=0}N^{(n)}$ where $N^{(n)} = \prod^{n}_{j=1}N$ and $N^{(0)} = \emptyset$.
\begin{definition}
A collection $\{G_{v}\}_{v\in \mathcal{V}}$ where $\mathcal{V} \subset V$ is a collection of finite sequences of players is called a strategic-form game with unawareness and the collection of views $\mathcal{V}$ is called its set of relevant views if the following properties are satisfied:
\begin{enumerate}
\item For every $v \in \mathcal{V}$,
\begin{equation}
v\string^\mathrm{v} \in \mathcal{V} ~\mbox{if and only if}~\mathrm{v} \in N_{v}.
\end{equation}
\item For every $v\string^\tilde{v} \in \mathcal{V}$,
\begin{equation}\label{formula19}
v \in \mathcal{V}, \quad \emptyset \ne N_{v\string^\tilde{v}} \subset N_{v}, \quad \emptyset \ne (A_{i})_{v\string^\tilde{v}} \subset (A_{i})_{v} ~\mbox{for all}~i \in N_{v\string^\tilde{v}}
\end{equation}
\item If $v\string^\mathrm{v}\string^\bar{v} \in \mathcal{V}$, then 
\begin{equation}\label{condition37}
v\string^\mathrm{v}\string^\mathrm{v}\string^\bar{v} \in \mathcal{V} ~\mbox{and}~ G_{v\string^\mathrm{v}\string^\bar{v}} = G_{v\string^\mathrm{v}\string^\mathrm{v}\string^\bar{v}}. 
\end{equation}
\item For every strategy profile $(s)_{v\string^\tilde{v}} = \{s_{j}\}_{j\in N_{v\string^\tilde{v}}}$, there exists a completion to an strategy profile $(s)_{v} = \{s_{j}, s_{k}\}_{j\in N_{v\string^\tilde{v}}, k\in N_{v}\setminus N_{v\string^\tilde{v}}}$ such that
\begin{equation}
(u_{i})_{{v\string^\tilde{v}}}((s)_{v\string^\tilde{v}}) = (u_{i})_{v}((s)_{v}).
\end{equation}
\end{enumerate}
\end{definition}
\subsection{Extended Nash equilibrium in strategic-form games with unawareness}\label{enesubsection}
In order to define extended Nash equilibrium it is needed to redefine the notion of strategy profile.
\begin{definition}
Let $\{G_{v}\}_{v\in \mathcal{V}}$ be a strategic-form game with unawareness. An extended strategy profile (ESP) in this game is a collection of strategy (pure or mixed) profiles $\{(\sigma)_{v}\}_{v\in \mathcal{V}}$ where $(\sigma)_{v}$ is a strategy profile in the game $G_{v}$ such that for every $v\string^\mathrm{v}\string^\bar{v} \in \mathcal{V}$ holds
\begin{equation}\label{condition39}
(\sigma_{\mathrm{v}})_{v} = (\sigma_{\mathrm{v}})_{v\string^\mathrm{v}} ~\mbox{as well as}~ (\sigma)_{v\string^\mathrm{v}\string^\bar{v}} = (\sigma)_{v\string^\mathrm{v}\string^\mathrm{v} \string^ \bar{v}}.
\end{equation}
\end{definition}
To illustrate (\ref{condition39}) let us take the game $G_{12}$--the game that player 1 thinks that player 2 is considering. If player 1 assumes that player 2 plays strategy $(\sigma_{2})_{12}$ in the game $G_{12}$, she must assume the same strategy in the game $G_{1}$ that she considers, i.e., $(\sigma_{2})_{1} = (\sigma_{2})_{12}$. In other words, player 1 finds that player 2 is considering strategy $(\sigma_{2})_{12}$. Thus, player 1 considers that strategy in her game $G_{1}$.
Next step is to extend rationalizability from strategic-form games to the games with unawareness.
\begin{definition}\label{rationalizable}
An ESP $\{(\sigma)_{v}\}_{v\in \mathcal{V}}$ in a game with unawareness is called extended rationalizable if for every $v\string^\mathrm{v} \in \mathcal{V}$ strategy $(\sigma_{\mathrm{v}})_{v}$ is a best reply to $(\sigma_{-\mathrm{v}})_{v\string^\mathrm{v}}$ in the game $G_{v\string^\mathrm{v}}$.
\end{definition}
Consider a strategic-form game with unawareness $\{G_{v}\}_{v\in \mathcal{V}}$. For every relevant view $v \in \mathcal{V}$ the relevant views as seen from $v$ are defined to be $\mathcal{V}^{v} = \{\tilde{v} \in \mathcal{V}\colon v\string^\tilde{v} \in \mathcal{V}\}$. For $\tilde{v} \in \mathcal{V}^{v}$ define the game $G^{v}_{\tilde{v}} = G_{v\string^\tilde{v}}$. Then the game with unawareness as seen from $v$ is defined as $\{G^{v}_{\tilde{v}}\}_{\tilde{v} \in \mathcal{V}^{v}}$.

We are now in a position to define the counterpart of Nash equilibrium in games with unawareness.
\begin{definition}
An ESP $\{(\sigma)_{v}\}_{v\in \mathcal{V}}$ in a game with unawareness is called an extended Nash equilibrium (ENE) if it is rationalizable and for all $v, \bar{v} \in \mathcal{V}$ such that $\{G^{v}_{\tilde{v}}\}_{\tilde{v} \in \mathcal{V}^{v}} = \{G^{\bar{v}}_{\tilde{v}}\}_{\tilde{v} \in \mathcal{V}^{\bar{v}}}$ we have that $(\sigma)_{v} = (\sigma)_{\bar{v}}$.
\end{definition}
The first part of the definition (rationalizability) is similar to the standard Nash equilibrium where it is required that each strategy in the equilibrium is a best reply to the other strategies of that profile. According to Definition~\ref{rationalizable}, player 2's strategy $(\sigma_{2})_{1}$ in the game of player 1 has to be a best reply to player 1's strategy $(\sigma_{1})_{12}$ in the game $G_{12}$. On the other hand, in contrast to the concept of Nash equilibrium, $(\sigma_{1})_{12}$ does not have to a best reply to $(\sigma_{2})_{1}$ but to strategy $(\sigma_{2})_{121}$. 

We saw in (\ref{formula11}) of Example 1 that for $v \in \{21, 121, 212, 1212, \dots\}$ we have $G_{v} = \Gamma_{2}$. It follows that $\{G_{21\string^v}\}_{v\in \mathcal{V}_{0}}  = \{G_{121\string^v}\}_{v\in \mathcal{V}_{0}} = \{\Gamma_{2}\}$. According to the second part of ENE, $(\sigma)_{21} = (\sigma)_{121}$.

The following proposition is useful to determine the extended Nash equilibria.
\begin{proposition}\label{proposition1}
Let $G$ be a strategic-form game and $\{G_{v}\}_{v \in \mathcal{V}}$ a strategic-form game with unawareness such that for some $v\in \mathcal{V}$ we have $G_{v\string^\bar{v}} = G$ for every $\bar{v}$ such that $v\string^\bar{v} \in \mathcal{V}$. Let $\sigma$ be a strategy profile in $G$. Then
\begin{enumerate}
\item $\sigma$ is rationalizable for $G$ if and only if $(\sigma)_{v} = \sigma$ is part of an extended rationalizable profile in $\{G_{v}\}_{v\in \mathcal{V}}$.
\item $\sigma$ is a Nash equilibrium for $G$ if and only if $(\sigma)_{v} = \sigma$ is part of on an ENE for $\{G_{v}\}_{v\in \mathcal{V}}$ and this ENE also satisfies $(\sigma)_{v} = (\sigma)_{v\string^\bar{v}}$.
\end{enumerate}
\end{proposition}
\begin{remark}
We see from (\ref{condition37}) and (\ref{condition39}) that for every $v\string^\mathrm{v}\string^\bar{v} \in \mathcal{V}$ a normal-form game $G_{v\string^\mathrm{v}\string^\bar{v}}$ and a strategy profile $(\sigma)_{v\string^\mathrm{v}\string^\bar{v}}$ determine the games and profiles in the form $G_{v\string^\mathrm{v}\string^\dots \string^\mathrm{v}\string^\bar{v}}$ and $(\sigma)_{v\string^\mathrm{v}\string^\dots \string^\mathrm{v}\string^\bar{v}}$, respectively. Hence, in general, a game with unawareness $\{G_{v}\}_{v\in \mathcal{V}}$ and an extended strategy profile $\{(\sigma)_{v}\}_{v\in \mathcal{V}}$ are defined by $\{G_{v}\}_{v\in \mathcal{V}_{0}}$ and $\{(\sigma)_{v}\}_{v\in \mathcal{V}_{0}}$, where
\begin{equation}
\mathcal{V}_{0} = \{v\in V \mid v=(i_{1}, \dots, i_{n}) ~\mbox{with}~i_{k} \ne i_{k+1} ~\mbox{for all}~k\}.
\end{equation}
Then, we get $\{G_{v}\}_{v\in \mathcal{V}}$ from $\{G_{v}\}_{v\in \mathcal{V}_{0}}$ by setting $G_{\tilde{v}} = G_{v}$ for $v=(i_{1},\dots, i_{n})\in \mathcal{V}_{0}$ and $\tilde{v} = (i_{1}, \dots, i_{k}, i_{k}, i_{k+1}, \dots, i_{n}) \in \mathcal{V}$. For this reason, we restrict ourselves to $\mathcal{V}_{0}$ throughout the paper.
\end{remark}
\section{Quantum PQ Penny Flip game with unawareness}\label{section4}
We noted in Subsection~\ref{techsection} that a single bimatrix game does not properly reflect the PQ Penny Flip game adjusted so that Alice can win every time  \cite{meyer}.
We now show that the problem may be regarded as a game with unawareness. 
\begin{example}
Following the description of the game given in Section~\ref{sec:1}, we may assume that the game $\Gamma_{QQ}$ is viewed by the modeler. Next, Alice (player 1)  being aware of quantum strategies may view her game $G_{1}$ as $\Gamma_{QC}$ or $\Gamma_{QQ}$, while Bob (player 2) perceives his game $G_{2}$ as $\Gamma_{CC}$. We can then assume that Alice finds that Bob is considering the classical PQ Penny Flip game, i.e., $G_{12} = \Gamma_{CC}$. Since $G_{2} = G_{12} = \Gamma_{CC}$, it follows from (\ref{formula19}) that any higher order views $v \in \{21, 121, 212, \dots \}$ are associated with  $\Gamma_{CC}$. We thus obtain a game with unawareness $\{G_{v}\}_{v\in \mathcal{V}_{0}}$, where
\begin{equation}
G_{v} = \begin{cases} \Gamma_{QQ} &\mbox{if}~v=\emptyset, \\ \Gamma_{QC} &\mbox{if}~v=1, \\ \Gamma_{CC} &\mbox{otherwise},
\end{cases} ~\mbox{or}~G'_{v} = \begin{cases} \Gamma_{QQ} &\mbox{if}~v=\emptyset, \\ \Gamma_{QQ} &\mbox{if}~v=1, \\ \Gamma_{CC} &\mbox{otherwise}.
\end{cases}
\end{equation}
Having defined a game with unawareness we are able to specify the players' optimal moves by using the notion of extended Nash equilibrium. We let $(\sigma^{c}_{1}, \sigma^{c}_{2})$ stand for a pair of optimal strategies in $\Gamma_{CC}$ (the optimal strategies of game (\ref{classic}) with $I_{j}$ and $X_{k}$ replaced by $\mathds{1}$ and $\sigma_{x}$, respectively). 
By Proposition~\ref{proposition1}, the strategy profile $(\sigma^{c}_{1}, \sigma^{c}_{2})$ is part of ENE for $v\in \mathcal{V}_{0}\setminus\{\emptyset, 1\}$. To determine $(\sigma)_{1} = (\sigma_{1}, \sigma_{2})_{1}$, first note that by the definition of extended strategy profile,
\begin{equation}
(\sigma_{2})_{1} = (\sigma_{2})_{12} = \sigma^c_{2}.
\end{equation}
According to Definition~\ref{rationalizable}, $(\sigma_{1})_{1}$ is a best reply to $(\sigma_{2})_{1} = \sigma^c_{2}$ in game $G_{1} = \Gamma_{QC}$ or $G_{1} = \Gamma_{QQ}$. Since Alice has access to unitary strategies in either case, she has a strategy guaranteeing a payoff of 1. One of the optimal strategies is playing the Hadamard matrix $H$ twice (see Eq.~\ref{mixedaw}).
Thus we can set $(\sigma_{1})_{1} = HH$. 

Determining $(\sigma)_{2} = (\sigma_{1}, \sigma_{2})_{2}$ runs along similar lines. We have $(\sigma_{1})_{2} = (\sigma_{1})_{21} = \sigma^c_{1}$, and $(\sigma_{2})_{2} = \sigma^{c}_{2}$ is a best reply to $(\sigma_{1})_{2} = \sigma^{c}_{1}$ in the game $\Gamma_{CC}$. Finally, (\ref{condition39}) implies that 
\begin{equation}
(\sigma_{1})_{\emptyset} = (\sigma_{1})_{1} \quad \text{and} \quad (\sigma_{2})_{\emptyset} = (\sigma_{2})_{2}.
\end{equation}
In summary, a possible extended Nash equilibrium is of the form
\begin{equation}
(\sigma)_{v} = \begin{cases}
(HH, \sigma^{c}_{2}) &\mbox{if}~ v\in \{\emptyset, 1\},\\
(\sigma^{c}_{1}, \sigma^{c}_{2}) &\mbox{otherwise}.
\end{cases}
\end{equation}
In words, Bob is only aware of $\mathds{1}$ and $\sigma_{x}$ and therefore he preceives a rational result of the game as $(\sigma)_{2} = (\sigma^{c}_{1}, \sigma^c_{2})$. Alice is fully aware of her unitary strategies in the modeler's game $G_{\emptyset}$. Thus, her prediction about a rational strategy profile coincides with the actual final result 
\begin{equation}
(\sigma)_{\emptyset} = ((\sigma_{1})_{1}, (\sigma_{2})_{2}) = (HH, \sigma^{c}_{2}). 
\end{equation}
Applying an ordinary Nash equilibrium to game $\Gamma_{QC}$ would lead us to an incorrect prediction about Bob's optimal strategy in the quantum Penny Flip game. Since $HH$ is Alice's winning strategy in $\Gamma_{QC}$ (see, Eq. (\ref{mixedaw})), every probability distribution over $\{\mathds{1}, \sigma_{x}\}$ is Bob's optimal strategy rather that a single strategy $\sigma^{c}_{2}$.
\end{example}
Our next example concerns a modification of quantum PQ Penny Flip game with finite strategy sets. 
\begin{example}\label{eexample1}
Consider strategic-form game (\ref{qq}), where we choose the following strategy sets:
\begin{equation}\label{finitestrategysets}
S_{1} = \{\mathds{1}, \sigma_{x}, H\}^2, \quad S_{2} = \{\mathds{1}, \sigma_{x}, \sigma_{z}\}.
\end{equation}
Then, according to scheme (\ref{qq})-(\ref{hermitianp}), the matrix representation of the game and its reduced form take on the form
\begin{equation}\label{gameexample1}
\bordermatrix{ & \mathds{1} & \sigma_{x} & \sigma_{z} \cr 
\mathds{1}\mathds{1} & 1 & -1 & 1 \cr
\mathds{1}\sigma_{x} & -1 & 1 & -1 \cr 
\mathds{1}H & 0 & 0 & 0 \cr 
\sigma_{x}\mathds{1} & -1 & 1 & -1 \cr
\sigma_{x}\sigma_{x} & 1 & -1 & 1 \cr 
\sigma_{x}H & 0 & 0 & 0 \cr 
H\mathds{1} & 0 & 0 & 0 \cr
H\sigma_{x} & 0 & 0 & 0 \cr
HH & 1 & 1 & -1}, \qquad 
\bordermatrix{& \mathds{1} & \sigma_{x} & \sigma_{z} \cr
s_{1} & 1 & -1 & 1 \cr 
s_{2} & -1 & 1 & -1 \cr
s_{3} & 0 & 0 & 0 \cr
s_{4}  & 1 & 1 & -1}.
\end{equation}
An easy computation shows that the value of the games of (\ref{gameexample1}) is 0. Player 1's optimal strategies in the reduced form are $(1/2, 1/2, 0, 0)$, $(0, 0, 1, 0)$ and $(1/2,0,0,1/2)$ (and any probability distribution over these strategies). Player 2's optimal strategy is $(0, 1/2, 1/2)$. The result so obtained is valid because we tacitly assume that the form of the game is common knowledge among the players.

Let us now modify the game defined by (\ref{qq}) and (\ref{finitestrategysets}) in how player 1 perceives player 2's perception of the game. Suppose that player 1 is unaware that player 2 is aware of actions $H$ and $\sigma_{z}$. On the other hand, we assume that player 2 considers game (\ref{gameexample1}). Furthermore, he knows how player 1 perceives her perception of the game. We can describe this problem formally as a strategic-form game with unawareness $\{G_{v}\}_{v\in \mathcal{V}_{0}}$, where the strategy sets of the players in each $G_{v}$, 
\begin{align}\label{gameexample3}
G_{v} = \left(\{1,2\}, \{S_{1}, S_{2}\}_{v}, \{\mathrm{tr}(\rho_{\mathrm{f}}P), -\mathrm{tr}(\rho_{\mathrm{f}}P)\} \right) 
\end{align}
are as follows:
\begin{equation}\label{esyexample1}
\{S_{1}, S_{2}\}_{v} = \begin{cases}\{\{\mathds{1}, \sigma_{x}, H\}^2, \{\mathds{1}, \sigma_{x}, \sigma_{z}\}\}, & \mbox{if}~v\in \{\emptyset, 1, 2, 21\}\\
\{\{\mathds{1}, \sigma_{x}\}^2, \{\mathds{1}, \sigma_{x}\}\} & \mbox{otherwise}.
\end{cases}
\end{equation}
Let us determine an ENE in the above game. Note first that the game
\begin{equation}
\Gamma_{CC} = \left(\{1,2\}, \{\{\mathds{1}\mathds{1}, \mathds{1}\sigma_{x}, \sigma_{x}\mathds{1}, \sigma_{x}\sigma_{x}\}, \{\mathds{1}, \sigma_{x}\}\}, \{\mathrm{tr}(\rho_{\mathrm{f}}P), -\mathrm{tr}(\rho_{\mathrm{f}}P)\}\right)
\end{equation}
satisfies the assumption of Proposition~\ref{proposition1} for $v=12$ and $v=212$, i.e.,
\begin{equation}\label{conditionene}
G_{12\string^\bar{v}} = G_{212\string^\bar{v}} = \Gamma_{CC}
\end{equation}
for every $\bar{v}$ such that $v\string^\bar{v} \in \mathcal{V}_{0}$. As a result, Nash equilibria in $\Gamma_{CC}$ are part of ENE in a game $\{G_{v}\}_{v\in \mathcal{V}_{0}}$ given by (\ref{gameexample3}) and (\ref{esyexample1}). The matrix forms of $\Gamma_{CC}$  and (\ref{classic}) coincide and so do the optimal strategies. Recall that $(\sigma^{c}_{1}, \sigma^{c}_{2})$ denotes a pair of optimal strategies in $\Gamma_{CC}$.  It follows that $(\sigma)_{v} = (\sigma^{c}_{1}, \sigma^{c}_{2})$
is part of the ENE for $v \in \{12,121,212,\dots\}$. We will now use the notion of extended rationalizability (see, Definition~\ref{rationalizable}) to determine the other profiles of $(\sigma)_{v}$. First, it must be the case that $(\sigma_{2})_{21}$ is a best reply to $(\sigma_{1})_{212} = \sigma^{c}_{1}$ in $G_{212}$. Hence $(\sigma_{2})_{21} = \sigma^{c}_{2}$. Next, $(\sigma_{1})_{21}$ has to be a best reply to $(\sigma_{2})_{211}$ in $G_{211}$. But $(\sigma_{2})_{211} = (\sigma_{2})_{21}$ and $G_{211} = G_{21}$ by Eq.~(\ref{condition37}) and (\ref{condition39}). Therefore, $(\sigma_{1})_{21}$ is part of the ENE if $(\sigma_{1})_{21}$ is a best reply to $(\sigma_{2})_{21} = \sigma^{c}_{2}$ in $G_{21}$. We thus get $(\sigma_{1})_{21} = HH$. As a result, $(\sigma)_{21} = (HH, \sigma^{c}_{2})$. Let us now find $(\sigma)_{2}$. In this case, $(\sigma_{1})_{2} = HH$ is a best reply to $(\sigma_{2})_{21} = \sigma^{c}_{2}$ in $G_{21}$. On the other hand, $(\sigma_{2})_{2}$ that is a best reply to $(\sigma_{1})_{2}$ in $G_{2}$ is $\sigma_{z}$. This gives $(\sigma)_{2} = (HH, \sigma_{z})$. We conclude similarly that $(\sigma)_{1} = (HH, \sigma^{c}_{2})$ and $(\sigma)_{\emptyset} = (HH, \sigma_{z})$. To sum up, the ENE is of the form
\begin{equation}
(\sigma)_{v} = \begin{cases}(\sigma^{c}_{1}, \sigma^{c}_{2}) &\mbox{if}~v\in \{12,121, 212, \dots\}, \\
(HH, \sigma^{c}_{2}) &\mbox{if}~v\in \{1,21\}, \\
(HH, \sigma_{z}) &\mbox{if}~v\in \{\emptyset, 2\}.
\end{cases}
\end{equation}
The ENE predicts that the game with unawareness ends with the payoff result of -1 determined by $(\sigma)_{\emptyset} = (HH, \sigma_{z})$.
\end{example}
The above example shows that incomplete awareness may dramatically affect the result of the game. In what follows we shall show that this is also true in a general setting, where the set of available actions for the players is the set of $2\times 2$ unitary matrices $U(2)$. 
\subsection{Relevant best replies in PQ Penny Flip-type games}
Recall that the unitary matrix $R_{\vec{n}}(\theta)$ corresponding to counterclockwise rotation through an angle $\theta$ about the axis directed along the unit vector $\vec{n} = (n_{x}, n_{y}, n_{z})$ is given by
\begin{equation}
R_{\vec{n}}(\theta) = \cos\frac{\theta}{2}\mathds{1} - i\sin\frac{\theta}{2}(n_{x}\sigma_{x} + n_{y}\sigma_{y} + n_{z}\sigma_{z}).
\end{equation}
In particular, the rotation matrices about the $x$, $y$, and $z$ axes are
\begin{align}\begin{split}
&R_{x}(\theta) = \begin{pmatrix} \cos\frac{\theta}{2} & -i\sin\frac{\theta}{2} \\ -i\sin\frac{\theta}{2} & \cos\frac{\theta}{2}\end{pmatrix},\\
&R_{y}(\theta) = \begin{pmatrix}  \cos\frac{\theta}{2} & -\sin\frac{\theta}{2} \\ \sin\frac{\theta}{2}  & \cos\frac{\theta}{2}\end{pmatrix}, \\
&R_{z}(\theta) = \begin{pmatrix}  e^{-i\theta/2} & 0 \\ 0 & e^{i\theta/2}\end{pmatrix}.
\end{split}
\end{align}
In order state our results we need to apply the following proposition \cite{shende}.
\begin{proposition}\label{lematnielsen}
Let $\vec{m}, \vec{n} \in \mathds{R}^3$ be unit vectors, $\vec{m} \bot \vec{n}$, and $U\in \mathsf{SU}(2)$. Then one can find real numbers $\beta$, $\gamma$, and $\delta$ such that 
\begin{equation}
U = R_{\vec{n}}(\beta)R_{\vec{m}}(\gamma)R_{\vec{n}}(\delta).
\end{equation}
\end{proposition}
We are now in a position to prove the lemmas that determine players' best replies to specific strategies.
Let $|\pm\rangle = (|0\rangle \pm |1\rangle)/\sqrt{2}$. The following lemma is a reformulation of the results appeared in \cite{iqbalgeometric} and \cite{classicalbaka}.
\begin{lemma}\label{lemma2}
The optimal strategy for player 1 in game $\Gamma_{QC}$ is a pair of unitary matrices $(V_{1},V_{3})$ such that 
\begin{equation}\label{a}
V_{1}|0\rangle \langle 0|V^{\dagger}_{1} \in \{|+\rangle \langle +|, |-\rangle \langle -|\}, \quad V_{3} = R_{z}(\alpha)V^{\dagger}_{1}
\end{equation}
for $\alpha \in \mathds{R}$. The matrix representation of $V_{1}$ (up to the global phase factor) is
\begin{equation}\label{lemma1}
V_{1} = \frac{R_{z}(a)}{\sqrt{2}}\left(\begin{array}{rr} \mathrm{e}^{-\mathrm{i}\gamma/2} & -\mathrm{i}\mathrm{e}^{\mathrm{i}\gamma/2}\\ \mathrm{e}^{-\mathrm{i}\gamma/2} & \mathrm{i}\mathrm{e}^{\mathrm{i}\gamma/2}\end{array}\right)
\end{equation}
for $a \in \{-\pi,0\}$ and $\gamma \in \mathds{R}$.
\end{lemma}
\begin{proof}
Note first that according to the definition of $\Gamma_{QC}$, a mixed strategy of player 2 is represented by a probability distribution $(p,1-p)$ over $\{\mathds{1}, \sigma_{x}\}$. Let $\rho$ be a state corresponding to a result of playing a mixed strategy $(p,1-p)$ by player 2 against a strategy $(U_{1}, U_{3})$ chosen by player 1. Then $\rho$ may be written as
\begin{equation}
\rho = pU_{3}U_{1}|0\rangle \langle 0|U^{\dagger}_{1}U^{\dagger}_{3} + (1-p)U_{3}\sigma_{x}U_{1}|0\rangle \langle 0|U^{\dagger}_{1}\sigma_{x}U^{\dagger}_{3}.
\end{equation}
Let $(V_{1}, V_{3})$ be a strategy of player 1 such that $\mathrm{tr}(\rho P) = 1$ for every mixed strategy of player 2. This clearly forces 
\begin{equation}\label{equalities}
\begin{cases}
V_{3}V_{1}|0\rangle \langle 0|V^{\dagger}_{1}V^{\dagger}_{3} = |0\rangle \langle 0|, \cr
V_{3}\sigma_{x}V_{1}|0\rangle \langle 0|V^{\dagger}_{1}\sigma_{x}V^{\dagger}_{3} = |0\rangle \langle 0|.
\end{cases}
\end{equation}
Combining Eqs.~(\ref{equalities}) we obtain
\begin{equation}
\sigma_{x}V_{1}|0\rangle \langle 0|V^{\dagger}_{1}\sigma_{x} = V_{1}|0\rangle \langle 0|V^{\dagger}_{1}.
\end{equation}
Since the eigenvectors of $\sigma_{x}$ are $|+\rangle$ and $|-\rangle$, and $V_{1}|0\rangle \langle 0|V^{\dagger}_{1} =|\Psi\rangle \langle \Psi|$, where $|\Psi\rangle$ is a unit vector, player 1's optimal action $V_{1}$ in game $\Gamma_{QC}$ satisfies either $V_{1}|0\rangle \langle 0|V^{\dagger}_{1} = |+\rangle \langle +|$ or  $V_{1}|0\rangle \langle 0|V^{\dagger}_{1} =|-\rangle \langle -|$.

In what follows, we derive the matrix representation of $V_{1}$. Let us first consider the case
\begin{equation}\label{1condition}
V_{1}|0\rangle \langle 0| V^{\dagger}_{1} = |+\rangle \langle +|.
\end{equation}
By Proposition~\ref{lematnielsen}, the matrix $V_{1}$ may be written (up to the global phase factor) as
\begin{equation}
V_{1} = \mathrm{e}^{\mathrm{i}\pi/4}R_{z}(\beta)R_{x}(\gamma)R_{z}(\delta).
\end{equation}
Let us determine $\beta$, $\gamma$ and $\delta$ so that equation~(\ref{1condition}) is satisfied. First, note that $R_{z}(\delta)$ has no effect on $|0\rangle \langle 0|$, i.e., $R_{z}(\delta)|0\rangle \langle 0|R^{\dagger}_{z}(\delta) = |0\rangle \langle 0|$. This is because, $R_{z}(\delta)$ corresponds to a counterclockwise rotation through an angle $\delta$ about the $z$-axis, and state $|0\rangle$ is represented by a point on that axis (see Fig~\ref{fig:2}). It follows that $\delta \in \mathds{R}$. We are left with the task of determining $\beta$, $\gamma$. We conclude from equation
\begin{equation}
R_{z}(\beta)R_{x}(\gamma)|0\rangle \langle 0|R^{\dagger}_{x}(\gamma)R^{\dagger}_{z}(\beta) = |+\rangle \langle +|
\end{equation}
that 
\begin{equation}
(\beta, \gamma) \in \{(\pi/2, \pi/2), (-\pi/2, -\pi/2)\}.
\end{equation}
\begin{figure}
  \includegraphics[width=0.9\textwidth]{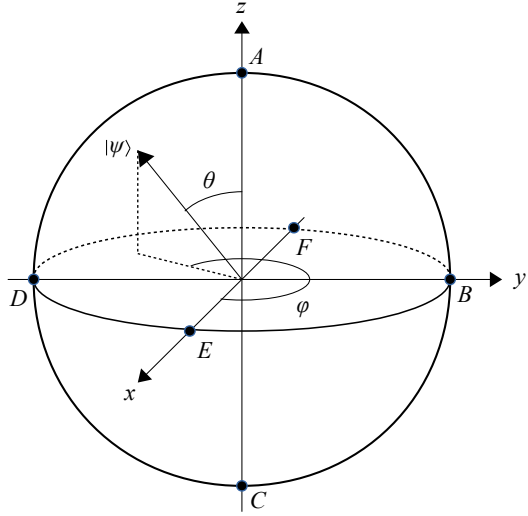}
\caption{Bloch sphere representation of a qubit. The points $A$, $B$, $C$, $D$, $E$ and $F$ correspond to $|0\rangle$, $(|0\rangle + i|1\rangle)/\sqrt{2}$, $|1\rangle$, $(|0\rangle - i|1\rangle)/\sqrt{2}$, $(|0\rangle + |1\rangle)/\sqrt{2}$ and $(|0\rangle -|1\rangle)/\sqrt{2}$, respectively.}
\label{fig:2}       
\end{figure}
Indeed, starting at point $A$ on the Bloch sphere (Fig.~\ref{fig:2}), we have to set $R_{x}(\pi/2)$ or $R_{x}(-\pi/2)$ in order to reach point $E$ by using the rotation matrix about $z$-axis. As a result, we obtain
\begin{align}
\mathrm{e}^{\mathrm{i}\pi/4}R_{z}(\pi/2)R_{x}(\pi/2)R_{z}(\gamma) &= \frac{1}{\sqrt{2}}\left(\begin{array}{rr} \mathrm{e}^{-\mathrm{i}\gamma/2} & -\mathrm{i}\mathrm{e}^{\mathrm{i}\gamma/2}\\ \mathrm{e}^{-\mathrm{i}\gamma/2} & \mathrm{i}\mathrm{e}^{\mathrm{i}\gamma/2}\end{array}\right) \nonumber\\ &=\mathrm{e}^{\mathrm{i}\pi/4}R_{z}(-\pi/2)R_{x}(-\pi/2)R_{z}(\gamma+\pi).
\end{align}
Applying similar reasoning to the case
\begin{equation}\label{2condition}
V_{1}|0\rangle \langle 0| V^{\dagger}_{1} = |-\rangle \langle -|
\end{equation}
leads to 
\begin{equation}
V_{1}= \mathrm{e}^{\mathrm{i}\pi/4}R_{z}(\pi/2)R_{x}(-\pi/2)R_{z}(\gamma)~\mbox{or}~V_{1}= \mathrm{e}^{\mathrm{i}\pi/4}R_{z}(-\pi/2)R_{x}(\pi/2)R_{z}(\gamma).
\end{equation}
One can check that both forms of $V_{1}$ are the same up to $\gamma \in \mathds{R}$. Therefore, equation~(\ref{2condition}) implies
\begin{align}
V_{1} &= \mathrm{e}^{\mathrm{i}\pi/4}R_{z}(-\pi/2)R_{x}(\pi/2)R_{z}(\gamma)\\
&=\frac{1}{\sqrt{2}}\left(\begin{array}{rr} \mathrm{i}\mathrm{e}^{-\mathrm{i}\gamma/2} & \mathrm{e}^{\mathrm{i}\gamma/2}\\ -\mathrm{i}\mathrm{e}^{-\mathrm{i}\gamma/2} & \mathrm{e}^{\mathrm{i}\gamma/2}\end{array}\right) = \frac{1}{\sqrt{2}}R_{z}(-\pi)\left(\begin{array}{rr} \mathrm{e}^{-\mathrm{i}\gamma/2} & -\mathrm{i}\mathrm{e}^{\mathrm{i}\gamma/2}\\ \mathrm{e}^{-\mathrm{i}\gamma/2} & \mathrm{i}\mathrm{e}^{\mathrm{i}\gamma/2}\end{array}\right).
\end{align}
We next turn to determining $V_{3}$. Without restriction of generality we can assume that $V_{1}$ is given by equation~(\ref{1condition}). We deduce from the system of equations~(\ref{equalities}) that 
\begin{equation}
V^{\dagger}_{3}|0\rangle \langle 0|V_{3} = |+\rangle \langle +| = V_{1}|0\rangle \langle 0|V^{\dagger}_{1}.
\end{equation}
Hence the optimal action $V_{3}$ has the form of $V^{\dagger}_{1}$ up to the composition with rotation about the $z$ axis. Thus the general form of $V_{3}$ may be written as $R_{z}(\alpha)V^{\dagger}_{1}$, where $\alpha \in \mathds{R}$.
\end{proof}
As it was mentioned in \cite{meyer}, player 2 being aware of his unitary strategies has a counterstrategy $V_{2}$ to player 1's optimal  strategy played in $\Gamma_{QC}$. The following lemma provides the general form of $V_{2}$.
\begin{lemma}\label{lemmav2}
Player 2's best reply in game $\Gamma_{QQ}$ to strategy $(V_{1}, V_{3})$ given by (\ref{a}) is a unitary matrix $V_{2}$ such that 
\begin{equation}\label{rlemma2}
V_{2}|+\rangle  \langle +|V^{\dagger}_{2} = |-\rangle \langle -|. 
\end{equation}
Its possible matrix representation is
\begin{equation}\label{plusminus0}
V_{2} = \mathrm{e}^{\mathrm{i}\alpha}\left(\begin{array}{ll}\mathrm{i}\cos\frac{\gamma}{2} & -\sin\frac{\gamma}{2}  \\ \sin\frac{\gamma}{2} & -\mathrm{i}\cos\frac{\gamma}{2} \end{array}\right), \gamma \in \mathds{R}.
\end{equation}
\end{lemma}
\begin{proof}
Let us assume that $(V_{1}, V_{3})$ satisfies $V_{1}|0\rangle \langle 0|V^{\dagger}_{1} = |+\rangle \langle +|$. 
It follows easily that $V^{\dagger}_{3}|1\rangle \langle 1|V_{3} = |-\rangle \langle -|$.
It was shown in Theorem 2 \cite{meyer} that there exists $V_{2} \in U(2)$ such that 
\begin{equation}\label{trm1}
\mathrm{tr}(\rho_{\mathrm{f}}P) = -1. 
\end{equation}
Then $V_{2}$ is player 2's best reply to $(V_{1}, V_{3})$. From~(\ref{trm1}) we obtain 
\begin{equation}
\rho_{\mathrm{f}} = V_{3}V_{2}V_{1}|0\rangle \langle 0|V^{\dagger}_{1}V^{\dagger}_{2}V^{\dagger}_{3} = |1\rangle \langle 1|.
\end{equation}
This clearly forces
\begin{equation}\label{p1}
V_{2}|+\rangle \langle +|V^{\dagger}_{2} = V^{\dagger}_{3}|1\rangle \langle 1|V_{3} = |-\rangle \langle -|.
\end{equation}
We can apply similar arguments again, with condition $V_{1}|0\rangle \langle 0|V^{\dagger}_{1} = |+\rangle \langle +|$ replaced by $V_{1}|0\rangle \langle 0|V^{\dagger}_{1} = |-\rangle \langle -|$ to obtain
\begin{equation}\label{p2}
V_{2}|-\rangle \langle -|V^{\dagger}_{2} = |+\rangle \langle +|.
\end{equation}
It is obvious that equation~(\ref{p2}) is equivalent to (\ref{p1}).

We will now derive the matrix representation of $V_{2}$. The method is similar to that in the proof of Lemma~\ref{lemma2}. By Proposition~\ref{lematnielsen}, we may write $V_{2}$ (up to the global phase factor) in the following form
\begin{equation}
V_{2} = R_{x}(\beta)R_{z}(\gamma)R_{x}(\delta).
\end{equation}
Now equation~(\ref{rlemma2}) becomes
\begin{equation}\label{equation10}
R_{z}(\gamma)R_{x}(\delta)|+\rangle \langle+|R^{\dagger}_{x}(\delta)R^{\dagger}_{z}(\gamma) = R^{\dagger}_{x}(\beta)|-\rangle \langle -| R_{x}(\beta).
\end{equation}
Since $R_{x}(\delta)$ and $R^{\dagger}_{x}(\beta)$ only affect the global phase factor of $|+\rangle$ and $|-\rangle$, we find that $\delta, \beta \in \mathds{R}$, and equation~(\ref{equation10}) reduces to
\begin{equation}
R_{z}(\gamma)|+\rangle \langle +|R^{\dagger}_{z}(\gamma) = |-\rangle \langle -|.
\end{equation}
It follows that $\gamma \in \{-\pi, \pi\}$. We thus obtain
\begin{equation}\label{plusminus}
V_{2} = \mathrm{e}^{\mathrm{i}\alpha}R_{x}(\beta)R_{z}(-\pi)R_{x}(\delta) = \left(\begin{array}{ll} \mathrm{i}\cos\frac{\beta - \gamma}{2} & -\sin\frac{\beta-\gamma}{2} \\ \sin\frac{\beta-\gamma}{2} & -\mathrm{i}\cos\frac{\beta-\gamma}{2}\end{array}\right).
\end{equation}
Since $\beta$ and $\gamma$ are real numbers, an equivalent formulation of (\ref{plusminus}) is (\ref{plusminus0}).
\end{proof}
The next lemma characterizes player 2's optimal unitary strategy in the game against player 1 equipped with the classical strategies.
\begin{lemma}\label{lemma4}
Player 2's optimal strategy $W_{2} \in U(2)$ in $\Gamma_{CQ}$ is of the form
\begin{equation}\label{u2cp}
W_{2} = \frac{\mathrm{e}^{\mathrm{i}\alpha}}{\sqrt{2}}\left(\begin{array}{ll}\mathrm{e}^{\mathrm{i}(-\beta/2 - \delta/2)} & -\mathrm{e}^{\mathrm{i}(-\beta/2 + \delta/2)} \\ \mathrm{e}^{\mathrm{i}(\beta/2 - \delta/2)} & \mathrm{e}^{\mathrm{i}(\beta/2 + \delta/2)}\end{array} \right).
\end{equation}
\end{lemma}
\begin{proof}
We first determine the general final state $\rho_{\mathrm{f}}$ of $\Gamma_{CQ}$ resulting from playing player 1's mixed strategy (probability distribution $(p_{1}, p_{2}, p_{3},1-p_{1} - p_{2} - p_{3})$ over  $\{\mathds{1}\mathds{1}, \mathds{1}\sigma_{x}, \sigma_{x}\mathds{1}, \sigma_{x}\sigma_{x}\}$) and player 2's unitary strategy $W_{2}$ written in the form 
\begin{equation}
W_{2} = \mathrm{e}^{\mathrm{i}\alpha}R_{z}(\beta)R_{y}(\gamma)R_{z}(\delta).
\end{equation}
We obtain
\begin{align}\label{rhofcq}
\rho_{\mathrm{f}} &= p_{1}\mathds{1}W_{2}\mathds{1}|0\rangle \langle 0|\mathds{1}W^{\dagger}_{2}\mathds{1} + p_{2}\sigma_{x}W_{2}\mathds{1}|0\rangle \langle 0|\mathds{1}W^{\dagger}_{2}\sigma_{x} \nonumber \\
&\quad + p_{3}\mathds{1}W_{2}\sigma_{x}|0\rangle \langle 0|\sigma_{x}W^{\dagger}_{2}\mathds{1} + (1-p_{1} - p_{2} - p_{3})\sigma_{x}W_{2}\sigma_{x}|0\rangle \langle 0|\sigma_{x}W^{\dagger}_{2}\sigma_{x}\nonumber\\
&=\left(\begin{array}{cc} \frac{1}{2}(1+ (1-2p_{2}-2p_{3})\cos\gamma) & \dots \\ \dots &  \frac{1}{2} + \left(-\frac{1}{2} + p_{2} + p_{3}\right)\cos\gamma\end{array}\right).
\end{align}
Therefore, the payoff outcome corresponding to (\ref{rhofcq}) depends only on $\gamma$, and it is equal to
\begin{equation}\label{outcomecq}
\mathrm{tr}(\rho_{\mathrm{f}}P) = \left(2\cos^2\frac{\gamma}{2} - 1\right)(1-2p_{2}-2p_{3}).
\end{equation}
One can check that~(\ref{outcomecq}) coincides with the outcome in $\Gamma_{CC}$ when player 1 uses her mixed strategy $(p_{1}, p_{2}, p_{3}, 1-p_{1} - p_{2} - p_{3})$, and player 2 plays $\mathds{1}$ and $\sigma_{x}$ according to the probability distribution $(\cos^2(\gamma/2), 1- \cos^2(\gamma/2))$. Since player 2's optimal strategy in $\Gamma_{CC}$ is $(q, 1-q) = (1/2, 1/2)$, the value of $\cos(\gamma/2)$ is either  $-1/\sqrt{2}$ or $1/\sqrt{2}$. We thus get
\begin{equation}
W_{2} = \frac{\mathrm{e}^{\mathrm{i}\alpha}}{\sqrt{2}}\left(\begin{array}{ll}\pm\mathrm{e}^{\mathrm{i}(-\beta/2 - \delta/2)} & -\mathrm{e}^{\mathrm{i}(-\beta/2 + \delta/2)} \\ \mathrm{e}^{\mathrm{i}(\beta/2 - \delta/2)} & \pm\mathrm{e}^{\mathrm{i}(\beta/2 + \delta/2)}\end{array} \right).
\end{equation}
Note that the signs associated with the diagonal entries depend on whether we set $\beta$ and $\delta$ or $\beta + \pi$ and $\delta + \pi$. For this reason, the form of $W_{2}$ is (\ref{u2cp}).
\end{proof}
\subsection{Examples of extended Nash equilibria}
Having determined the relevant best replies in $\Gamma_{QQ}$, $\Gamma_{QC}$ and $\Gamma_{CQ}$ we are now in a position to study the quantum PQ Penny Flip game with unawareness. Consider a family of games $\{G_{v}\}_{v\in \mathcal{V}_{0}}$ where
\begin{equation}\label{qpfg1}
G_{v} = \begin{cases} \Gamma_{QQ} &\mbox{if}~v\in\{\emptyset, 1, 2, 21\}, \\ \Gamma_{CC} &\mbox{if}~v\in\{12, 121, 212, \dots\}. \end{cases}
\end{equation}
Game (\ref{qpfg1}) generalizes the game defined by (\ref{esyexample1}). The set of actions available to the players is now the set of all $2\times 2$ unitary matrices. At the same time, the game has the same structure of unawareness as the game in Example~1. This fact implies that both games have the same structure of ENE.
\begin{proposition}\label{propositionv1v2v3}
Let $\{G_{v}\}_{v\in \mathcal{V}}$ be a game with unawareness defined by (\ref{qpfg1}). The set of extended Nash equilibria is given by the following formula
\begin{equation}\label{proposition1ene}
(\sigma)_{v} = \begin{cases} (\sigma^c_{1}, \sigma^c_{2}) &\mbox{if}~v\in \{12, 121, 212, \dots\},\\
((V_{1}, V_{3}), \sigma^c_{2}) &\mbox{if}~v\in \{1, 21\},\\
((V_{1}, V_{3}), V_{2}) &\mbox{if}~v \in \{\emptyset, 2\}.  \end{cases}
\end{equation}
\end{proposition}
\begin{proof}
As in Example~\ref{eexample1}, the game $G = \Gamma_{CC}$ meets condition~(\ref{conditionene}). This fact justifies the first piece of (\ref{proposition1ene}). Let us justify $(\sigma)_{21}$. From (\ref{condition39}) we obtain $(\sigma_{2})_{21} = (\sigma_{2})_{212} = \sigma^c_{2}$. Turning to $(\sigma_{1})_{21}$, by Definition~\ref{rationalizable}, we need to determine player 1's best response to $(\sigma_{2})_{21} = \sigma^c_{2}$ in $G_{21} = \Gamma_{QQ}$. By Lemma~\ref{lemma2}, player 1's optimal strategy to any probability mixture over $\mathds{1}$ and $\sigma_{x}$ is $(V_{1}, V_{3})$. Let us examine the strategy profile $(\sigma)_{2}$. Again, we see from (\ref{condition39}) that $(\sigma_{1})_{2} = (\sigma_{1})_{21} = (V_{1}, V_{3})$. On the other hand, it follows from Lemma~\ref{lemmav2} that $V_{2}$ is player 2's best reply to $(V_{1}, V_{3})$ in $G_{2} = \Gamma_{QQ}$. Similar reasoning applies to the other profiles of (\ref{proposition1ene}).
\end{proof}
We conclude from Proposition~\ref{propositionv1v2v3} that $\{G_{v}\}_{v\in \mathcal{V}_{0}}$ given by (\ref{qpfg1}) favors Player 2. The ENE generates the best possible payoff for player 2,
\begin{equation}
\mathrm{tr}\left(\left(V_{3}V_{2}V_{1}|0\rangle \langle 0| V^{\dagger}_{1}V^{\dagger}_{2}V^{\dagger}_{3}\right)P\right) = -1.
\end{equation}
We now investigate the case where each player perceives that the other player plays $\Gamma_{CC}$. Although, a two-component strategy set may seem to be to player 2's advantage in any zero-sum game (see, Theorem 5.44 in \cite{maschler}), the corresponding ENE does not prejudge the outcome; each player has the chance of getting her most preferred outcome. To be specific, consider $\{G_{v}\}$ defined as follows:
\begin{equation}\label{cqfg1}
G_{v} = \begin{cases}  \Gamma_{QQ} &\mbox{if}~v\in \{\emptyset, 1, 2\}, \\ 
\Gamma_{CC} &\mbox{if}~\mbox{otherwise}.
\end{cases}
\end{equation}
We can formulate the following proposition:
\begin{proposition}
Let $\{G_{v}\}_{v\in \mathcal{V}_{0}}$ be a game with unawareness defined by (\ref{cqfg1}). The set of extended Nash equilibria is given by the following formula
\begin{equation}\label{profilesl4}
(\sigma)_{v} = \begin{cases} ((V_{1}, V_{3}), W_{2}) &\mbox{if}~ v = \emptyset, \\
((V_{1}, V_{3}), \sigma^c_{2}) &\mbox{if}~ v = 1,\\ (\sigma^c_{1}, W) &\mbox{if}~ v = 2, \\ (\sigma^{c}_{2}, \sigma^{c}_{2}) &\mbox{if}~\mbox{otherwise}.\end{cases}
\end{equation}
\end{proposition}
\begin{proof}
Analysis similar to that in the proof of Proposition~\ref{propositionv1v2v3} shows
that $(\sigma)_{v} = (\sigma^c_{1}, \sigma^c_{2})$ for $v\in \{12, 21, 121, 212, \dots\}$ and $(\sigma_{1})_{2} = \sigma^c_{1}$. By Lemma~\ref{lemma4}, player 2's best reply to $\sigma^c_{1}$ is given by (\ref{u2cp}). We thus obtain $(\sigma)_{2} = (\sigma^c_{1}, W)$. We leave it to the reader to verify the other profiles of (\ref{profilesl4}). \end{proof}
The ENE predicts $(\sigma)_{\emptyset} = ((V_{1}, V_{3}), W_{2})$ in game $\{G_{v}\}$ given by (\ref{cqfg1}). An easy computation shows that 
\begin{equation}\label{zero}
\mathrm{tr}(V_{3}W_{2}V_{1}|0\rangle \langle 0|V^{\dagger}_{1}W^{\dagger}_{2}V^{\dagger}_{3}) = -\sin\beta_{2}\sin\delta_{2}.
\end{equation}
According to (\ref{profilesl4}), player 2 predicts that the result of the game is $(\sigma)_{2} = (\sigma^{c}_{1}, W)$, and so player 2 does not have most-preferred parameters $\beta$ and $\delta$ in $W_{2}$. If we assume that $(\beta, \delta)$ are uniformly distributed over $[0, 2\pi] \times [0, 2\pi]$ then the expected value of (\ref{zero}) is equal to 0. 
\section{Conclusions}
We have shown that the notion of game with unawareness is a necessary tool in studying the quantum PQ Penny Flip game. Different players' perceptions of strategies available in the game require using more sophisticated methods for describing the game and its possible rational results than an ordinary matrix game together with the concept of Nash equilibrium. The examples used in the paper indicate that not only the possibility of using quantum strategies but also incomplete awareness of the players may lead to unpredictable outcomes. This fact undoubtedly sheds new light on quantum game theory.

Our work provides new tools that might be utilized in allied sciences. The obtained results can be generalized to more complex games, and then applied to study numerous economical problems formulated in terms of games with unawareness with the use of mathematical methods of quantum information. At the same time these problems will enrich theory of quantum information through new examples that will show superiority of using quantum methods over methods of classical information theory.




\section*{Acknowledges}
\noindent This work was supported by the National Science Centre, Poland under the research project 2016/23/D/ST1/01557.




\end{document}